\documentclass[12pt]{article}
\usepackage[margin=1.2in]{geometry}

\usepackage{url,amsmath,amssymb,amsthm,natbib,appendix,bbm,mathtools,booktabs,caption,makecell,float,enumerate,scalerel}
\usepackage[dvipsnames]{xcolor}
\usepackage{pgf,tikz,pgfplots,mathrsfs,graphicx}
\usepackage{hyperref,cleveref}
\usetikzlibrary{arrows}
\theoremstyle{plain}
\newtheorem{definition}{Definition}
\newtheorem{thm}{Theorem}
\newtheorem{prop}{Proposition}

\newtheorem{axm}{Axiom}
\theoremstyle{definition}

\theoremstyle{remark}

\usepackage{chngcntr,apptools}
\AtAppendix{\counterwithin{lem}{section}}
\AtAppendix{\counterwithin{prop}{section}}
\AtAppendix{\counterwithin{axm}{section}}
\title{Random Non-Expected Utility: Non-Uniqueness}
\author{Yi-Hsuan Lin\thanks{Institute of Economics, Academia Sinica; \href{mailto:yihsuanl@econ.sinica.edu.tw}{yihsuanl@econ.sinica.edu.tw}. I started this project during my PhD studies at Boston University. I am deeply grateful to Larry Epstein for his advice and encouragement. I also thank Jay Lu for his helpful comments and for raising the question addressed in Section 3.2.}}
\date{August 2020\footnote{First draft: October 2018. \href{https://drive.google.com/file/d/1yBVWLw_Q4JcPSvTYs1cLvEKADIiiomkX/view?usp=sharing}{{\it Please click here for the latest version}}}}
\begin{document}
\maketitle\begin{NoHyper}
\abstract{In random expected utility \citep*{Gul:2006}, the distribution of preferences is uniquely recoverable from random choice. This paper shows through two examples that such uniqueness fails in general if risk preferences are random but do not conform to expected utility theory. In the first, non-uniqueness obtains even if all preferences are confined to the betweenness class \citep*{Dekel:1986} and are suitably monotone. The second example illustrates random choice behavior consistent with random expected utility that is also consistent with random non-expected utility. On the other hand, we find that if risk preferences conform to weighted utility theory  \citep*{Chew:1983} and are monotone in first-order stochastic dominance, random choice again uniquely identifies the distribution of preferences. Finally, we argue that, depending on the domain of risk preferences, uniqueness may be restored if joint distributions of choice across a limited number of feasible sets are available.}
\section{Introduction}
A classic model of random choice is random utility theory. It can be interpreted as follows. In a heterogeneous population, suppose each individual maximizes her preference.  Given any feasible set $D$ of alternatives, as different individual might make different choice, the choice behavior of this population is summarized by a distribution over $D$, which is determined by the distribution of preferences in the population.\footnote{Random utility theory can also model the stochastic behavior of a single agent. It hypothesizes that the agent's preference is random according to a fixed distribution. Facing a feasible set of alternatives, she first perceives the realized preference and then makes a rational choice. Hence, from ex-ante point of view, her choice appears random.} In the context of choice under risk, a special case of the model is random expected utility (REU). All individuals in the population are expected utility maximizers, but their risk attitudes are not identical. Under REU, the distribution of preferences is uniquely recoverable from random choice \citep*{Gul:2006}. In other words, when an analyst observes only the choice frequencies for this population then, under the assumption that each individual's preference conforms to expected utility theory, she is able to identify a unique distribution of preferences consistent with the observed behavior.
\par
The focus on expected utility preferences is natural as a first step, but is not completely satisfactory in light of its well-known descriptive failures, such as the Allais paradox. One might suspect that the observed random choice of lotteries could be rationalized by random non-expected utility but not by REU. On the other hand, in an abstract choice setting, it is well-known that two distinct random preferences can rationalize the same random choice behavior.\footnote{See \cite{Barbera:1986} for an example of random choice that can be rationalized by more than one random utility.} Hence if risk preferences are unrestricted at all, random choice of lotteries does not identify the underlying distribution of preferences uniquely.
\par
Nonetheless, even if a modeler wants to deviate from expected utility theory, she does not need to embrace all kinds of risk preferences. Many classes of non-expected utility preferences have certain structures for tractability and are consistent with many stylized behaviors violating expected utility theory. Therefore, a natural question is if the uniqueness result of REU remains true when we make a small deviation from expected utility? Moreover, if a class of non-expected utility does not yield unique identification, how can we restore uniqueness by enriching the observable behavior? This paper aims to respond to these questions.
\par
Firstly, we show by example that random choice may not identify a unique distribution of preferences when risk preferences are random and restricted but do not conform to expected utility theory. In particular, we have three findings for non-uniqueness: (i) non-uniqueness obtains even if risk preferences are confined to the betweenness (implicit expected utility) class, developed by \cite{Dekel:1986} and \citet{Chew:1983,Chew:1989}, and are all monotonic with respect to first-order stochastic dominance; (ii) non-uniqueness obtains even if risk preferences are confined to the weighted utility class \citep*{Chew:1983}; (iii) random choice may be rationalized by both REU and random non-expected utility.  Thus, non-unique identification seems to be a generic problem for random non-expected utility models.
\par
The third finding above demonstrates a subtle difference between the expected utility hypothesis at the individual and population levels. {\it Even if the observed choice frequencies from a population can be rationalized by random expected utility, it is still possible that no individual in the population is an expected utility maximizer.}  Therefore, while violating any axiom of REU implies that not all individuals are expected utility maximizers, consistency with REU does not validate the expected utility hypothesis at the individual level either.
\par
We then deepen our analysis under weighted utility theory and provide some positive results on restoring unique identification. While random weighted utility is not uniquely identified from random choice in general, we may regain uniqueness by requiring all preferences to be monotonic with respect to first-order stochastic dominance. We provide a formal argument for the case of three prizes. In some choice settings, for example, when the prizes are monetary, such monotonicity is normative appealing. Therefore, if unique identification is desirable for a modeler, random weighted utility with monotonicity in stochastic dominance may be a good choice for random non-expected utility models. 
\par
Without such monotonicity, we can still uniquely identify random weighted utility if we suitably enrich the observable behavior. In the case of three prizes, we show that the distribution of weighted utility preferences is uniquely recoverable from joint distribution of choice across three binary menus. Note that under the classic notion of random choice, only the distribution of choice from each feasible set is observable. {\it The joint distribution of choice across any two sets is not.} Under REU, random choice implicitly reveals all joint choice probabilities and so pins down a unique distribution of preferences. Once we deviate from expected utility, random choice no longer discloses such information. Thus, non-uniqueness obtains in general. If we can observe joint choice frequencies across choice sets, we will be able to improve identification. Our finding suggests that when risk preferences are sufficiently restricted, joint choice distributions across a small number of menus suffice for unique identification. 
\par
To the best of our knowledge, this is the first study to address the identification of random non-expected utility from a decision-theoretic perspective.\footnote{\cite{Lu:2020} studies a random utility model where each individual has an ambiguity-averse preference over Anscombe-Aumann acts. Thus, the population's choice of acts is not consistent with random subjective expected utility. Nonetheless, he considers the set of all lotteries over acts as the choice domain and assumes that each individual is still an expected-utility maximizer when evaluating lotteries. Hence his identification result essentially follows from the uniqueness of REU.} The paper does not aim to dive deep into any particular model of random non-expected utility. Instead, it is mainly illustrative. We argue that some classes of non-expected utility preferences do not admit unique identification of random utility, and some do. We suggest that allowing joint choice probabilities as the observable can help to regain uniqueness. To convey these points easily, we conduct all the analyses in a three-prize setting so that a risk preference can be visualized on a two-dimensional plane. While risk preferences in this paper all belong to the betweenness class or a subclass, we acknowledge the existence of many other types of non-expected utility widely used in economics but not considered here yet. Future work may focus on a particular class of risk preferences and perform a more detailed analysis. 
\par
The rest of this paper is organized as follows. Section 2 introduces random implicit expected utility (RIEU). Section 3 provides two examples of two distinct RIEUs that induce the same random choice. Hence the distribution of risk preferences is not uniquely identified. Section 4 points out a class of non-expected utility that yields unique identification of random utility. It also discusses the reason for general non-uniqueness and demonstrates that uniqueness may be restored if choice data are suitably enriched. In the appendix, we point out three behavioral properties of RIEU.

\section{Random Utility Model}
\subsection{Basic modeling approach of random choice and random utility}
We first review the general random utility model. There is a universal space of choice alternatives $X$. A choice set, or called a menu, is a finite subset of $X$. Let $\mathcal{D}$ denote the collection of all menus.
\par
Choice from a menu is modeled as a random set. In particular, for any $A,D\in\mathcal{D}$, we denote by $\rho_D(A)$ the probability that $A$ is the set of all optimal alternatives in $D$. To ensure feasibility, we require that if $\rho_D(A)>0$ then $A$ is a nonempty subset of $D$. Let $\Pi$ be the set of all simple probability measures over the class of all finite subsets of $X$. The observable choice behavior is summarized by a random choice correspondence.\footnote{We consider multi-valued random choice to avoid dealing with ties in choice. Literature often studies single-valued random choice. Namely, the choice out of $A$ is summarized by a distribution over $A$, but not over the class of subsets of $A$. With single-valued random choice, a random utility model must assume that ties occur with zero probability, or impose a tie-breaking rule. Otherwise, the connection between the observable and the model would be loose.}
\begin{definition}
A random choice correspondence (RCC) is a function $\rho:\mathcal{D}\rightarrow\Pi$ with $\rho_D(\{A\in\mathcal{D}:\emptyset\neq A\subset D\})=1$ for all $D\in\mathcal{D}$.
\end{definition}
\par
A preference relation $\succsim$ is a complete and transitive binary relation over $X$. For each menu $D\in\mathcal{D}$, let $M(D,\succsim)$ denote the set of all optimal options in $D$ according to $\succsim$; that is, $$M(D,\succsim)\coloneqq\{x\in D:x\succsim y\ \forall\ y\in D\}.$$
\par
Fix a set of preference relations $\Omega$. Let $N(D,A)$ denote the set of preferences in $\Omega$ under which $A$ is the set of all optimal options in $D$; that is, $$N(D,A)\coloneqq\{\succsim\in\Omega:A=M(D,\succsim)\}\subset\Omega.$$ Let \begin{equation}\mathcal{C}\coloneqq\{N(D,A):A,D\in\mathcal{D}\},\label{C}\end{equation} and let $\mathcal{F}(\mathcal{C})$ denote the smallest field that contains every element of $\mathcal{C}$. 
\par
A random utility, or called a random preference, is a finitely additive probability measure $\mu$ on $(\Omega,\mathcal{F}(\mathcal{C}))$. Say that $\rho$ is rationalized by $\mu$ if $\rho_D(A)$ equals the probability under $\mu$ that $A$ is the set of all optimal alternatives in $D$.
\begin{definition}
Random choice correspondence $\rho$ is rationalized by random utility $\mu$ if, for all $D,A\in\mathcal{D}$, $$\rho_D(A)=\mu(N(D,A)).$$
\end{definition}

\subsection{Random utility models for choice under risk and identification problem}
This paper exclusively focuses on random choice under risk. There is a finite set of prizes denoted $W=\{w_1,w_2,\cdots,w_{N+1}\}$ for $N\geq 1$. The objects of choice are lotteries over $W$. Let $\Delta\coloneqq\{p\in\mathbb{R}_+^N:\sum_{n=1}^Np^n\leq 1\}$ be the set of all lotteries. For each $p\in\Delta$, its $n$th coordinate $p^n$ is the probability of winning the prize $w_n$, for all $n=1,\cdots,N$, and $p^{N+1}\coloneqq1-\sum_{n=1}^Np^n$ is the probability of winning the prize $w_{N+1}$. A lottery assigning probability one to the prize $w\in W$ is denoted by $w$.
\par
Let $X=\Delta$, and then we can define a random risk preference $\mu$ as before. The identification problem concerns if random choice of lotteries $\rho$ identifies a unique distribution of risk preferences $\mu$. That is, if $\rho$ is rationalizable by both $\mu$ and $\mu'$, is it necessary that $\mu=\mu'$? The answer to this question crucially depends on the domain of preferences $\Omega$.
\par
In order to facilitate the analysis, we assume that there are only three prizes ($N=2$) throughout the paper. Thus a lottery is identified as a point in a probability simplex in $\mathbb{R}^2$, and a risk preference can be described by its indifference map. For an expected-utility preference, the indifference sets are all linear and parallel to each other. A well-known generalization of expected utility is betweenness preference, which relaxes the parallelism of indifference curves. 
\begin{definition}[\citealp{Dekel:1986}]
A binary relation $\succsim$ over $\Delta$ is called a betweenness preference if it satisfies: 
\begin{enumerate}
\item $\succsim$ is complete and transitive.
\item There exist best and worst elements in $\Delta$ which are the sure prizes in $W$.
\item If $p\succ q\succ r$, then there exists $\alpha\in(0,1)$ such that $\alpha p+(1-\alpha)q\sim r$.
\item (Betweenness)\\If $p\succ q$, then $p\succ\alpha p+(1-\alpha)q\succ q$ for all $\alpha\in(0,1)$.\\ If $p\sim q$, then $p\sim\alpha p+(1-\alpha)q\sim q$ for all $\alpha\in(0,1)$.
\end{enumerate}
\end{definition}
A notable feature of such preference is that for any $p\in\Delta$, the indifference set $\{q\in\Delta:q\sim p\}$ is the intersection of a hyperplane and $\Delta$. Moreover, this hyperplane divides $\Delta$ into the upper and the lower contour sets of $p$ (i.e.\ $\{q\in\Delta:q\succ p\}$ and $\{q\in\Delta:q\prec p\}$). However, two indifference sets need not be parallel (i.e.\ their corresponding hyperplanes can intersect outside the simplex). If we strengthen Betweenness to Independence: $$\forall\alpha\in(0,1],\ \forall p,q,r,\in\Delta,\ p\succsim q\Leftrightarrow\alpha p+(1-\alpha)r\succsim\alpha q+(1-\alpha)r,$$ then $\succsim$ becomes an expected-utility preference whose indifference sets are all parallel to each other. See Figure~\ref{BP} as examples of betweenness preference and expected-utility preference in the Marschak-Machina triangle.
\par
\begin{figure}
\centering
\begin{tikzpicture}[line cap=round,line join=round,>=triangle 45,x=1.4cm,y=1.2cm]
\clip(-0.5,-0.5) rectangle (8.7,3.4);
\draw [line width=1pt,color=darkgray] (0,0)-- (2,3);
\draw [line width=1pt,color=darkgray] (2,3)-- (4,0);
\draw [line width=1pt,color=darkgray] (4,0)-- (0,0);
\draw [line width=1pt,color=darkgray] (4.5,0)-- (6.5,3);
\draw [line width=1pt,color=darkgray] (6.5,3)-- (8.5,0);
\draw [line width=1pt,color=darkgray] (8.5,0)-- (4.5,0);
\draw [line width=1pt,color=red] (0.16813131103702372,-0.2398955201720981)-- (3.5721554620877836,2.059631432931349);
\draw [line width=1pt,color=red] (4.685705915755984,-0.22631563659077825)-- (8.089730066806737,2.0732113165126638);
\draw [line width=1pt,color=red] (0.16813131103702372,-0.2398955201720981)-- (2.6668298899998164,2.783891890601726);
\draw [line width=1pt,color=red] (4.2330431297120015,0.36214598526640296)-- (7.637067280762754,2.6616729383698434);
\draw [line width=1pt,color=red] (5.265114281892286,-0.6880316783556436)-- (8.669138432943026,1.611495274747797);
\draw [line width=1pt,color=red] (0.5212082841513314,-0.1674694744050603)-- (3.9071259237603333,1.1452526051224978);
\draw [line width=1pt,color=red] (3.9071259237603333,1.1452526051224978)-- (2.087421523863517,-0.34853458882265453);
\draw [->,line width=1pt,color=gray] (3.02628725813012,1.1551342349840245) -- (2.3960605925140726,2.0242467841884769);
\draw [->,line width=1pt,color=gray] (7.483990307648455,1.1325010956818249) -- (6.877422174349515,2.0016136448862774);
\draw [line width=1pt,color=red] (1.1368296731711498,2.231643291628064)-- (2.2956464054437493,2.8382114249270045);
\draw (0,3) node[anchor=north west] {(a)};
\draw (4.5,3) node[anchor=north west] {(b)};
\begin{scriptsize}
\draw (1.3,1.1) node[anchor=north west] {$I(p)$};
\draw (5.8,1.1) node[anchor=north west] {$I(p)$};
\draw [fill=blue] (0,0) circle (2.5pt);
\draw[color=blue] (-0.15,-0.15) node {$w$};
\draw [fill=blue] (4,0) circle (2.5pt);
\draw[color=blue] (4,-0.2) node {$\underline{w}$};
\draw [fill=blue] (2,3) circle (2.5pt);
\draw[color=blue] (2,3.2) node {$\bar{w}$};
\draw [fill=blue] (4.5,0) circle (2.5pt);
\draw[color=blue] (4.35,-0.15) node {$w$};
\draw [fill=blue] (8.5,0) circle (2.5pt);
\draw[color=blue] (8.5,-0.2) node {$\underline{w}$};
\draw [fill=blue] (6.5,3) circle (2.5pt);
\draw[color=blue] (6.5,3.2) node {$\bar{w}$};
\draw [fill=blue] (2,1) circle (2.5pt);
\draw[color=blue] (2,1.2) node {$p$};
\draw [fill=blue] (6.5,1) circle (2.5pt);
\draw[color=blue] (6.5,1.2) node {$p$};
\end{scriptsize}
\end{tikzpicture}
\caption{}
\medskip
\begin{minipage}{0.75\textwidth}
{\scriptsize A stylized betweenness preference is depicted in (a). The indifference set containing $p$, denoted $I(p)$, is the intersection of a straight line and the simplex. Two indifference sets can be non-parallel. A stylized expected-utility preference is depicted in (b), where all indifference sets are parallel. (Arrows indicate the direction of increasing preference.)\par}
\end{minipage}
\label{BP}
\end{figure}
\par
Each betweenness preference has an implicit expected utility representation which we state in Appendix A. Thus, if $\Omega$ is the set of all betweenness preferences, we call $\mu$ random implicit expected utility (RIEU).
\par
All random preferences considered in this paper are RIEUs. We can further restrict $\Omega$ by considering special cases of betweenness preferences. We will also consider monotonicity in first-order stochastic dominance. It requires an exogenous and fixed ranking of prizes. And then for every possible preference $\succsim$, $p\succsim q$ if  $p$ is obtained from $q$ by shifting probability mass from a worse prize to a better prize (i.e.\ $p$ first-order stochastic dominates $q$).

\section{Non-Uniqueness of RIEU}
In general, an RCC may be rationalized by more than one RIEU. We show this through two examples. In particular, the first example shows that non-uniqueness obtains even if risk preferences are all monotonic with respect to first-order stochastic dominance. The second example illustrates random choice consistent with random expected utility that can also be rationalized by more than one random weighted utility, a special case of RIEU.
\subsection{Example 1}
We define two RIEUs $\mu$ and $\mu'$ as follows. Under $\mu$, the realized preference is either $\succsim_1$ or $\succsim_2$ with equal probability. Figure~\ref{RIEU1} depicts their indifference maps. Under $\mu'$, the realized preference is either $\succsim'_1$ or $\succsim'_2$ with equal probability. Figure~\ref{RIEU2} depicts their indifference maps.

\begin{figure}
\centering
\begin{tikzpicture}[line cap=round,line join=round,>=triangle 45,x=3.5cm,y=3.5cm]
\clip(-0.6,-0.6) rectangle (3.1,1.25);
\draw [line width=2pt,color=darkgray] (0,0)-- (1,0);
\draw [line width=2pt,color=darkgray] (1,0)-- (0,1);
\draw [line width=2pt,color=darkgray] (0,1)-- (0,0);
\draw [line width=1pt,dashed,color=red] (-0.5,-0.5)-- (1,1);
\draw [line width=1pt,dashed,color=red] (-0.5,-0.5)-- (1.15,0.05);
\draw [line width=1pt,dashed,color=red] (-0.5,-0.5)-- (0.05,1.15);
\draw [line width=1pt,dashed,color=red] (-0.5,-0.5)-- (0.875,0.325);
\draw [line width=1pt,dashed,color=red] (-0.5,-0.5)-- (0.325,0.875);
\draw [->,line width=1pt,color=gray]  (-0.05,-0.35) -- (-0.35,-0.05);
\draw [line width=2pt,color=darkgray] (2,0)-- (3,0);
\draw [line width=2pt,color=darkgray] (3,0)-- (2,1);
\draw [line width=2pt,color=darkgray] (2,1)-- (2,0);
\draw [line width=1pt,dashed,color=red] (1.5,-0.5)-- (3,1);
\draw [line width=1pt,dashed,color=red] (3,1)-- (3,-0.2);
\draw [line width=1pt,dashed,color=red] (3,1)-- (2.6,-0.2);
\draw [line width=1pt,dashed,color=red] (3,1)-- (1.8,0.6);
\draw [line width=1pt,dashed,color=red] (3,1)-- (1.8,1);
\draw [->,line width=1pt,color=gray] (3,0.6) -- (2.6,1);
\draw [line width=1pt,color=red] (0,0)-- (0.5,0.5);
\draw [line width=1pt,color=red] (0.33333,0)-- (0.75,0.25);
\draw [line width=1pt,color=red] (0,0.33333)-- (0.25,0.75);
\draw [line width=1pt,color=red] (2,0)-- (2.5,0.5);
\draw [line width=1pt,color=red] (2.66666,0)-- (2.75,0.25);
\draw [line width=1pt,color=red] (2,0.66666)-- (2.25,0.75);
\draw (-0.5,1) node[anchor=north west] {(a)};
\draw (1.5,1) node[anchor=north west] {(b)};
\begin{scriptsize}
\draw [fill=blue] (1,0) circle (2pt);
\draw[color=blue] (1.05,0.05) node {$w_1$};
\draw [fill=blue] (0,0) circle (2pt);
\draw[color=blue] (-0.07,0) node {$w_3$};
\draw [fill=blue] (0,1) circle (2pt);
\draw[color=blue] (-0.05,1.05) node {$w_2$};
\draw [fill=gray] (-0.5,-0.5) circle (2pt);
\draw[color=darkgray] (-0.54,-0.54) node {$x_1$};
\draw [fill=gray] (1,1) circle (2pt);
\draw[color=darkgray] (1.04,1.04) node {$x_2$};
\draw [fill=blue] (0.75,0.25) circle (2pt);
\draw[color=blue] (0.77,0.33) node {$q'$};
\draw [fill=blue] (0.5,0) circle (2pt);
\draw[color=blue] (0.5,-0.06) node {$p'$};
\draw [fill=blue] (0.25,0.75) circle (2pt);
\draw[color=blue] (0.33,0.77) node {$q$};
\draw [fill=blue] (0,0.5) circle (2pt);
\draw[color=blue] (-0.06,0.5) node {$p$};
\draw [fill=blue] (3,0) circle (2pt);
\draw[color=blue] (3.05,0.05) node {$w_1$};
\draw [fill=blue] (2,0) circle (2pt);
\draw[color=blue] (1.93,0) node {$w_3$};
\draw [fill=blue] (2,1) circle (2pt);
\draw[color=blue] (1.95,1.05) node {$w_2$};
\draw [fill=gray] (1.5,-0.5) circle (2pt);
\draw[color=darkgray] (1.46,-0.54) node {$x_1$};
\draw [fill=gray] (3,1) circle (2pt);
\draw[color=darkgray] (3.04,1.04) node {$x_2$};
\draw [fill=blue] (2.75,0.25) circle (2pt);
\draw[color=blue] (2.82,0.27) node {$q'$};
\draw [fill=blue] (2.5,0) circle (2pt);
\draw[color=blue] (2.5,-0.06) node {$p'$};
\draw [fill=blue] (2.25,0.75) circle (2pt);
\draw[color=blue] (2.27,0.82) node {$q$};
\draw [fill=blue] (2,0.5) circle (2pt);
\draw[color=blue] (1.94,0.5) node {$p$};
\end{scriptsize}
\end{tikzpicture}
\caption{}
\medskip 
\begin{minipage}{0.75\textwidth}
{\scriptsize (a): The preference $\succsim_1$ is represented by the weighted utility function $V_1$. The worst lottery is $w_1$ and the best lottery is $w_2$. All indifference curves intersect at $x_1=(-\frac{1}{2},-\frac{1}{2})$.\\(b): The preference $\succsim_2$ is represented by the weighted utility function $V_2$. The worst lottery is $w_1$ and the best lottery is $w_2$. All indifference curves intersect at $x_2=(1,1)$.\par}
\end{minipage}
\label{RIEU1}

\centering
\begin{tikzpicture}[line cap=round,line join=round,>=triangle 45,x=3.5cm,y=3.5cm]
\clip(-0.6,-0.6) rectangle (3.1,1.25);
\draw [line width=2pt,color=darkgray] (0,0)-- (1,0);
\draw [line width=2pt,color=darkgray] (1,0)-- (0,1);
\draw [line width=2pt,color=darkgray] (0,1)-- (0,0);
\draw [line width=1pt,dashed,color=red] (-0.5,-0.5)-- (1,1);
\draw [line width=1pt,dashed,color=red] (-0.5,-0.5)-- (0.05,1.15);
\draw [line width=1pt,dashed,color=red] (-0.5,-0.5)-- (0.325,0.875);
\draw [->,line width=1pt,color=gray] (-0.2,-0.2) -- (-0.35,-0.05);
\draw [line width=2pt,color=darkgray] (2,0)-- (3,0);
\draw [line width=2pt,color=darkgray] (3,0)-- (2,1);
\draw [line width=2pt,color=darkgray] (2,1)-- (2,0);
\draw [line width=1pt,dashed,color=red] (1.5,-0.5)-- (3,1);
\draw [line width=1pt,dashed,color=red] (3,1)-- (1.8,0.6);
\draw [line width=1pt,dashed,color=red] (3,1)-- (1.8,1);
\draw [->,line width=1pt,color=gray] (2.8,0.8) -- (2.6,1);
\draw [line width=1pt,dashed,color=red] (1,1)-- (1,-0.2);
\draw [line width=1pt,dashed,color=red] (1,1)-- (0.6,-0.2);
\draw [line width=1pt,dashed,color=red] (1.5,-0.5)-- (3.15,0.05);
\draw [line width=1pt,dashed,color=red] (1.5,-0.5)-- (2.875,0.325);
\draw [->,line width=1pt,color=gray] (1.95,-0.35) -- (1.8,-0.2);
\draw [->,line width=1pt,color=gray] (1,0.6) -- (0.8,0.8);
\draw [line width=1pt,color=red] (0,0)-- (0.5,0.5);
\draw [line width=1pt,color=red] (2.33333,0)-- (2.75,0.25);
\draw [line width=1pt,color=red] (0,0.33333)-- (0.25,0.75);
\draw [line width=1pt,color=red] (2,0)-- (2.5,0.5);
\draw [line width=1pt,color=red] (0.66666,0)-- (0.75,0.25);
\draw [line width=1pt,color=red] (2,0.66666)-- (2.25,0.75);
\draw (-0.5,1) node[anchor=north west] {(a)};
\draw (1.5,1) node[anchor=north west] {(b)};
\begin{scriptsize}
\draw [fill=blue] (1,0) circle (2pt);
\draw[color=blue] (1.05,0.05) node {$w_1$};
\draw [fill=blue] (0,0) circle (2pt);
\draw[color=blue] (-0.07,0) node {$w_3$};
\draw [fill=blue] (0,1) circle (2pt);
\draw[color=blue] (-0.05,1.05) node {$w_2$};
\draw [fill=gray] (-0.5,-0.5) circle (2pt);
\draw[color=darkgray] (-0.54,-0.54) node {$x_1$};
\draw [fill=gray] (1,1) circle (2pt);
\draw[color=darkgray] (1.04,1.04) node {$x_2$};
\draw [fill=blue] (0.75,0.25) circle (2pt);
\draw[color=blue] (0.82,0.27) node {$q'$};
\draw [fill=blue] (0.5,0) circle (2pt);
\draw[color=blue] (0.5,-0.06) node {$p'$};
\draw [fill=blue] (0.25,0.75) circle (2pt);
\draw[color=blue] (0.33,0.77) node {$q$};
\draw [fill=blue] (0,0.5) circle (2pt);
\draw[color=blue] (-0.06,0.5) node {$p$};
\draw [fill=blue] (3,0) circle (2pt);
\draw[color=blue] (3.05,0.05) node {$w_1$};
\draw [fill=blue] (2,0) circle (2pt);
\draw[color=blue] (1.93,0) node {$w_3$};
\draw [fill=blue] (2,1) circle (2pt);
\draw[color=blue] (1.95,1.05) node {$w_2$};
\draw [fill=gray] (1.5,-0.5) circle (2pt);
\draw[color=darkgray] (1.46,-0.54) node {$x_1$};
\draw [fill=gray] (3,1) circle (2pt);
\draw[color=darkgray] (3.04,1.04) node {$x_2$};
\draw [fill=blue] (2.75,0.25) circle (2pt);
\draw[color=blue] (2.77,0.33) node {$q'$};
\draw [fill=blue] (2.5,0) circle (2pt);
\draw[color=blue] (2.5,-0.06) node {$p'$};
\draw [fill=blue] (2.25,0.75) circle (2pt);
\draw[color=blue] (2.27,0.82) node {$q$};
\draw [fill=blue] (2,0.5) circle (2pt);
\draw[color=blue] (1.94,0.5) node {$p$};
\end{scriptsize}
\end{tikzpicture}
\caption{}
\medskip
\begin{minipage}{0.75\textwidth}
{\scriptsize (a): The preference $\succsim'_1$ is represented by the utility function $V'_1$. The worst lottery is $w_1$ and the best lottery is $w_2$. For lotteries which are better than $w_3$, their indifference curves intersect at $x_1=(-\frac{1}{2},-\frac{1}{2})$. For lotteries which are worse than $w_3$, their indifference curves intersect at $x_2=(1,1)$.\\(b): The preference $\succsim'_2$ is represented by the utility function $V'_2$. The worst lottery is $w_1$ and the best lottery is $w_2$. For lotteries which are better than $w_3$, their indifference curves intersect at $x_2=(1,1)$. For lotteries which are worse than $w_3$, their indifference curves intersect at $x_1=(-\frac{1}{2},-\frac{1}{2})$.\par}
\end{minipage}
\label{RIEU2}
\end{figure}
\par
Note that these four preferences have identical indifference set of $w_3$. Above that set, $\succsim_1$ and $\succsim'_1$ have the same indifference map, and so do $\succsim_2$ and $\succsim'_2$. Below that set, $\succsim_1$ and $\succsim'_2$ have the same indifference map, and so do $\succsim_2$ and $\succsim'_1$.
\par
Below we provide numerical representations of $\succsim_i$ and $\succsim'_i$, for $i\in\{1,2\}$. In particular, $\succsim_1$ and $\succsim_2$ both follow {\it weighted utility theory} \citep*{Chew:1983}, which is a special case of betweenness preference. Weighted utility preferences are represented by the function $$V(p)=\frac{\sum_{n=1}^{N+1}p^ng(w_n)u(w_n)}{\sum_{n=1}^{N+1}p^ng(w_n)},$$ where $u(\cdot)$ and $g(\cdot)$ are real-valued functions defined on $W$, and $g$ is non-zero and nonnegative (or nonpositive).
\par
We have assumed $N=2$. Let $u(\cdot)$ be such that $u(w_1)=0$, $u(w_2)=1$, and $u(w_3)=\frac{1}{2}$. Let $g_1(\cdot)$ be such that $g_1(w_1)=g_1(w_2)=1$ and $g_1(w_3)=\frac{1}{2}$, and let $g_2(\cdot)$ be such that $g_2(w_1)=g_2(w_2)=1$ and $g_2(w_3)=2$. Let $V_i(\cdot)$ be the weighted utility function defined by $u$ and $g_i$ for $i\in\{1,2\}$; that is, $$V_1(p)=\frac{p^2+[(1-p^1-p^2)\times\frac{1}{2}\times\frac{1}{2}]}{p^1+p^2+[(1-p^1-p^2)\times\frac{1}{2}]}\mbox{ and }V_2(p)=\frac{p^2+[(1-p^1-p^2)\times2\times\frac{1}{2}]}{p^1+p^2+[(1-p^1-p^2)\times2]}.$$
\par
Define utility function $V'_i(\cdot)$ for $i\in\{1,2\}$ such that $$V'_1(p)=\left\{\begin{array}{ll}V_1(p)&\mbox{if }V_1(p)\geq\frac{1}{2},\\V_2(p)&\mbox{otherwise,}\end{array}\right.$$ and $$V'_2(p)=\left\{\begin{array}{ll}V_2(p)&\mbox{if }V_2(p)\geq\frac{1}{2},\\V_1(p)&\mbox{otherwise.}\end{array}\right.$$
\par
For $i\in\{1,2\}$, $V_i$ represents $\succsim_i$, and $V'_i$ represents $\succsim'_i$. Obviously, these four preferences all are betweenness preferences.\footnote{In fact, $\succsim'_1$ and $\succsim'_2$ both follow {\it semi-weighted utility} \citep*{Chew:1989}. A key geometric feature of such preference is that, there are at most two points outside the Marschak-Machina triangle at which two indifference curves can intersect. Our example illustrates that non-uniqueness obtains even if we further confine risk preferences to the class of semi-weighted utility.} They also agree on the ranking of prizes: $w_1$ is the worst prize, and $w_2$ is the best prize. Moreover, if we define first-order stochastic dominance based on this ranking of prizes, then under all these preferences, $p$ is better than $q$ whenever $p$ dominates $q$.\footnote{In Figures \ref{RIEU1} and \ref{RIEU2}, preferences increase when we shift a lottery toward northwest. Therefore, they are monotonic with respect to first-order stochastic dominance.}
\par
RCC $\rho$ is rationalized by $\mu$ if and only if it is rationalized by $\mu'$. To see this, suppose that lottery $p$ is such that $V_1(p)\geq\frac{1}{2}$. Then $V_2(p)\geq\frac{1}{2}$. For any other lottery $q$, $p\succsim_1q\Leftrightarrow p\succsim'_1q$ and $p\succsim_2q\Leftrightarrow p\succsim'_2q$. Similarly, if $V_1(p)<\frac{1}{2}$, then $p\succsim_1q\Leftrightarrow p\succsim'_2q$ and $p\succsim_2q\Leftrightarrow p\succsim'_1q$. Since $\mu(\succsim_i)=\mu'(\succsim'_j)$ for all $i,j\in\{1,2\}$, the probability that $p$ is better than $q$ is the same under $\mu$ and $\mu'$. It is straightforward to extend the argument to show that, for any menu $D$ and $A\subset D$, the probability that $A$ is the set of all optimal lotteries in $D$ is the same under $\mu$ and $\mu'$.
\par
Therefore, RCC induced by $\mu$ can be rationalized by at least two RIEUs. We conclude that, under RIEU, random choice may not identify a unique distribution of preferences. It is true even if we impose monotonicity in stochastic dominance.
\begin{prop}
There exists a random choice correspondence rationalizable by more than one random implicit expected utility.
\end{prop}

\subsection{Example 2}
We will review in Section 4 that a random choice correspondence cannot be rationalized by two different random expected utilities (REU). However, it might be rationalizable by REU and also by random non-expected utility. The following example illustrates. In particular, it shows that random expected utility and random weighted utility (RWU) can induce identical random choice.
\par
When there are only three prizes, each weighted utility preference \citep*{Chew:1983} is characterized by (i) a point outside the Marschak-Machina triangle at which all indifference curves intersect, and (ii) the direction of increasing preference (either clockwise or counterclockwise).\footnote{Here, we exclude expected utility from the class of weighted utility instead of viewing it as a special case. That is, we mean ``strict'' weighted utility.} 
\par
Construct RWU $\nu_1$ as follows (see Figure~\ref{RIEU3}). Fix a circle surrounding the triangle. Let all possible intersection points of indifference curves be distributed uniformly on the circle. Moreover, conditional on any intersection point, let the direction of increasing preference be clockwise or counterclockwise with equal probability. 
\par
Figure~\ref{RIEU4} illustrates the implication of $\nu_1$ for random choice. Suppose that $\rho$ is rationalized by $\nu_1$. Take any $p,q,r\in\Delta$. Taking $p$ as the vertex, let $\alpha$ be the angle (in degrees) formed by these lotteries  ($\measuredangle qpr=\alpha^\circ$). It can be shown that $\rho_{\{p,q,r\}}(p)=\frac{1}{2}\left(1-\frac{\alpha}{180}\right)$. 
\par
Let $\nu_2$ be a uniform distribution over expected utility preferences. Then, under $\nu_2$, the probability that $p$ is optimal in $\{p,q,r\}$ is also $\frac{1}{2}\left(1-\frac{\alpha}{180}\right)$. In fact, for any $p\in\Delta$, $\nu_1$ and $\nu_2$ induce the identical distribution over the class of lower contour sets $\{p'\in\Delta:p\succsim p'\}$. Therefore, $\nu_1$ and $\nu_2$ induce the same random choice. 
\par
Note that when constructing $\nu_1$, the circle surrounding the triangle is chosen arbitrarily. By choosing a different circle, we can construct another RWU that rationalizes the same $\rho$. Thus, non-uniqueness obtains even if all preferences are confined to the (strict) weighted utility class.
\begin{prop}
There exists a random choice correspondence rationalizable by more than one random weighted utility.
There exists a random choice correspondence rationalizable by both random expected utility and random weighted utility.
\end{prop}
\par
Suppose that an analyst aims to test expected utility theory at the individual level, but only has choice data at the group level. Then she might consider testing REU theory instead. However, while the consistency between observed random choice and REU is necessary for all individuals to be expected-utility maximizers, it is not sufficient. Our example illustrates that, even if the observed behavior satisfies all the axioms of REU, it could be that no individual in the population is an expected-utility maximizer.
\par
This observation connects to classical demand theory. It is known that, even if each consumer behaves irrationally, aggregate demand could satisfy the weak axiom of revealed preference, even be consistent with maximization of a single preference (\citealp{Becker:1962}; \citealp{Grandmont:1992}). Our result here is an analogy if we view expected utility theory as a rational choice behavior under risk. The example demonstrates that if individuals violate the independence axiom in diverse directions, then their irrational behaviors ``cancel out''. Hence the observed population choice could be fully rational.

\begin{figure}
\centering
\begin{tikzpicture}[line cap=round,line join=round,>=triangle 45,x=3.5cm,y=3.5cm]
\clip(-0.5,-0.5) rectangle (3.5,1.5);
\draw [line width=2pt] (0,0)-- (1,0);
\draw [line width=2pt] (1,0)-- (0,1);
\draw [line width=2pt] (0,1)-- (0,0);
\draw [line width=1pt,color=red] (-0.13949708113416207,-0.1387785384861777)-- (1,1);
\draw [line width=1pt,color=red] (-0.13949708113416207,-0.1387785384861777)-- (1.1962635952886291,0.0247839933206939);
\draw [line width=1pt,color=red] (-0.13949708113416207,-0.1387785384861777)-- (0.034969619459835156,1.2678592350529179);
\draw [line width=1pt,color=red] (-0.13949708113416207,-0.1387785384861777)-- (1.1635510889272547,0.43369032283787284);
\draw [line width=1pt,color=red] (-0.13949708113416207,-0.1387785384861777)-- (0.45478011776414096,1.2133383911172941);
\draw [->,line width=1pt,color=darkgray] (0.2639571639894565,-0.07880561015699143) -- (-0.08497623719853792,0.2973882129988132);
\draw [line width=2pt] (2,0)-- (3,0);
\draw [line width=2pt] (3,0)-- (2,1);
\draw [line width=2pt] (2,1)-- (2,0);
\draw [line width=1pt,color=red] (1.9322949884195548,-0.10061394773124098)-- (2.9027660104736643,1.3060238258078547);
\draw [line width=1pt,color=red] (2.9027660104736643,1.3060238258078547)-- (3.017259782738475,-0.171491044847552);
\draw [line width=1pt,color=red] (2.9027660104736643,1.3060238258078547)-- (2.6,-0.2);
\draw [line width=1pt,color=red] (2.9027660104736643,1.3060238258078547)-- (1.8,0.6);
\draw [line width=1pt,color=red] (2.9027660104736643,1.3060238258078547)-- (1.7905407941869322,0.92983000265205);
\draw [->,line width=1pt,color=darkgray] (2.9409306012286014,0.8153362303872399) -- (2.510215934137171,1.169721715968795);
\draw [line width=1.2pt,dotted,color=ForestGreen] (0.5,0.5) circle (3.15cm);
\draw [line width=1.2pt,dotted,color=ForestGreen] (2.5,0.5) circle (3.15cm);
\draw (-0.5,1.4) node[anchor=north west] {(a)};
\draw (1.5,1.4) node[anchor=north west] {(b)};
\begin{scriptsize}
\draw [fill=blue] (1,0) circle (2pt);
\draw[color=blue] (1.05,0.05) node {$w_1$};
\draw [fill=blue] (0,0) circle (2pt);
\draw[color=blue] (-0.08,0) node {$w_3$};
\draw [fill=blue] (0,1) circle (2pt);
\draw[color=blue] (0.05,1.05) node {$w_2$};
\draw [fill=blue] (-0.13949708113416207,-0.1387785384861777) circle (2pt);
\draw[color=blue] (-0.2,-0.2) node {$x_1$};
\draw [fill=blue] (3,0) circle (2pt);
\draw[color=blue] (3.05,0.05) node {$w_1$};
\draw [fill=blue] (2,0) circle (2pt);
\draw[color=blue] (1.92,0) node {$w_3$};
\draw [fill=blue] (2,1) circle (2pt);
\draw[color=blue] (2.05,1.05) node {$w_2$};
\draw [fill=blue] (2.9027660104736643,1.3060238258078547) circle (2pt);
\draw[color=blue] (2.96,1.36) node {$x_2$};
\end{scriptsize}
\end{tikzpicture}
\caption{}
\medskip 
\begin{minipage}{0.75\textwidth} 
{\scriptsize (a): The preference follows weighted utility theory. All indifference curves intersect at $x_1$, and the preference increases counterclockwise.\\(b): The preference follows weighted utility theory. All indifference curves intersect at $x_2$, and the preference increases clockwise.\par}
\end{minipage}
\label{RIEU3}

\centering
\begin{tikzpicture}[line cap=round,line join=round,>=triangle 45,x=4cm,y=4cm]
\clip(-0.5,-0.5) rectangle (1.5,1.5);
\draw [shift={(0.39298982797043414,0.520923673134872)},line width=1pt,fill=black,fill opacity=0.1] (0,0) -- (-144.10733431225452:0.10793411920506535) arc (-144.10733431225452:-57.17616803289554:0.10793411920506535) -- cycle;
\draw [line width=2pt] (0,0)-- (1,0);
\draw [line width=2pt] (1,0)-- (0,1);
\draw [line width=2pt] (0,1)-- (0,0);
\draw [line width=1.2pt,dotted,color=ForestGreen] (0.5,0.5) circle (3.6cm);
\draw [line width=1.2pt,dotted,color=darkgray] (-0.2721644680441808,0.03750552357233953)-- (1.1817247035724638,1.0952098959234442);
\draw [line width=1.2pt,dotted,color=darkgray] (0.9236593756105095,-0.2968889858994874)-- (-0.0540810923016841,1.215155752581817);
\draw [line width=2pt,color=darkgray] (0.39298982797043414,0.520923673134872)-- (0.11675088536327166,0.3210139120375841);
\draw [line width=2pt,color=darkgray] (0.39298982797043414,0.520923673134872)-- (0.6,0.2);
\draw [shift={(0.5,0.5)},line width=2pt,color=ForestGreen]  plot[domain=-1.082156979789941:0.7177495978319877,variable=\t]({1*0.9025072422926036*cos(\t r)+0*0.9025072422926036*sin(\t r)},{0*0.9025072422926036*cos(\t r)+1*0.9025072422926036*sin(\t r)});
\draw [shift={(0.5,0.5)},line width=2pt,color=ForestGreen]  plot[domain=2.229962571552465:3.6812459978122996,variable=\t]({1*0.904684258344917*cos(\t r)+0*0.904684258344917*sin(\t r)},{0*0.904684258344917*cos(\t r)+1*0.904684258344917*sin(\t r)});
\begin{scriptsize}
\draw [fill=blue] (0.39298982797043414,0.520923673134872) circle (2pt);
\draw[color=blue] (0.34,0.52) node {$p$};
\draw [fill=blue] (0.11675088536327166,0.3210139120375841) circle (2pt);
\draw[color=blue] (0.12,0.38) node {$q$};
\draw [fill=blue] (0.6,0.2) circle (2pt);
\draw[color=blue] (0.62,0.25) node {$r$};
\draw [fill=blue] (-0.0540810923016841,1.215155752581817) circle (2.5pt);
\draw[color=blue] (-0.06,1.29) node {$D$};
\draw [fill=blue] (1.1817247035724638,1.0952098959234442) circle (2.5pt);
\draw[color=blue] (1.2,1.16) node {$A$};
\draw [fill=blue] (-0.2721644680441808,0.03750552357233953) circle (2.5pt);
\draw[color=blue] (-0.3,-0.02) node {$F$};
\draw [fill=blue] (0.9236593756105095,-0.2968889858994874) circle (2.5pt);
\draw[color=blue] (0.95,-0.35) node {$C$};
\draw [fill=blue] (1.4,0.5) circle (2.5pt);
\draw[color=blue] (1.46,0.5) node {$B$};
\draw [fill=blue] (-0.4,0.5) circle (2.5pt);
\draw[color=blue] (-0.46,0.5) node {$E$};
\draw[color=black] (0.38,0.38) node {$\alpha^\circ$};
\end{scriptsize}
\end{tikzpicture}
\caption{}
\medskip 
\begin{minipage}{0.75\textwidth} 
{\scriptsize Under $\nu_1$, the realized preference ranks $p$ optimal in $D\equiv\{p,q,r\}$ if and only if either (i) the indifference curves intersect at some point on arc $\widehat{ABC}$ and the preference increases clockwise, or (ii) the indifference curves intersect at some point on arc $\widehat{DEF}$ and the preference increases counterclockwise. Because all possible intersection points distribute uniformly on the circle, the realized point lies on $\widehat{ABC}\cup\widehat{DEF}$ with probability $1-\frac{\alpha}{180}$. Because the direction of increasing preference is clockwise or counterclockwise with equal probability, $\rho_D(p)=\frac{1}{2}\left(1-\frac{\alpha}{180}\right)$ if $\rho$ is rationalized by $\nu_1$.\par}
\end{minipage}
\label{RIEU4}
\end{figure}

\section{Restoring Uniqueness}
This section provides some positive results on the unique identification of random non-expected utility.  We find one class of non-EU preferences under which random choice identifies underlying distribution of preferences uniquely. Then we discuss the reason for general non-uniqueness and suggest how to enrich the observable to improve identification. 

\subsection{Random weighted utility with monotonicity in stochastic dominance}

\begin{figure}
\centering
\begin{tikzpicture}
[line cap=round,line join=round,>=triangle 45,x=2.5cm,y=2.5cm]
\clip(-3,-2) rectangle (3,2);
\draw[->,ultra thin,color=darkgray] (-3,0)--(3,0) node[right]{$r_1$};
\draw[->,ultra thin,color=darkgray] (0,-2)--(0,2) node[above]{$r_2$};
\draw [line width=2pt,color=darkgray] (-1,0)-- (0,1);
\draw [line width=2pt,color=darkgray] (0,1)-- (0,-1);
\draw [line width=2pt,color=darkgray] (0,-1)-- (-1,0);
\draw [line width=1pt,color=darkgray,dashed] (-4,-3)-- (2,3);
\draw [line width=1pt,color=darkgray,dashed] (-4,3)-- (2,-3);
\draw [line width=1pt,color=darkgray,dashed] (-2,-3)-- (4,3);
\draw [line width=1pt,color=darkgray,dashed] (-2,3)-- (4,-3);
\path [fill=ForestGreen,opacity=0.1] (-4,3)--(-4,-3)--(-1,0)--(-4,3);
\path [fill=blue,opacity=0.1] (4,3)--(4,-3)--(1,0)--(4,3);
\path [fill=darkgray,opacity=0.1] (-1,0)--(0,1)--(0,-1)--(-1,0);
\draw[red, semithick, dashed, domain=-3:3] plot (\x, {((1-(-0.5))/(0-(-2)))*(\x-0)+1});
\draw[red, semithick, dashed, domain=-3:3] plot (\x, {((-1-(-0.5))/(0-(-2)))*(\x-0)-1});
\draw[red, semithick, domain=-3:3] plot (\x, {((0.4-(-0.5))/(0-(-2)))*(\x-0)+0.4});
\draw[color=black] (0,1.12) node {$\bar{w}$};
\draw[color=black] (0,-1.12) node {$\underline{w}$};
\draw[color=black] (-1.12,0) node {$w$};
\draw[fill=black] (-2,-0.5) circle (2pt);
\draw[color=black] (-2,-0.6) node {$x$};
\draw[fill=red] (0,0.4) circle (2pt);
\draw[color=red] (0.14,0.33) node {$p_a$};
\draw[color=red] (1.8,0.75) node {$\tilde{m}_a=a\tilde{m}_1+(1-a)\tilde{m}_0$};
\draw[color=red] (1,1.6) node {$\tilde{m}_1$};
\draw[color=red] (2,-1.35) node {$\tilde{m}_0$};
\draw [decorate,red,decoration={brace,amplitude=5pt},xshift=2pt,yshift=0pt] (0,1) -- (0,0.4) node [red,midway,xshift=0.6cm] {\footnotesize $1-a$};
\draw [decorate,red,decoration={brace,amplitude=5pt},xshift=-2pt,yshift=0pt] (0,-1) -- (0,0.4) node [red,midway,xshift=-0.3cm] {\footnotesize $a$};
\end{tikzpicture}
\caption{}
\medskip
\begin{minipage}{0.75\textwidth} 
{\scriptsize The best and worst prizes, $\bar{w}$ and $\underline{w}$ , are identified with the points $(0,1)$ and $(0,-1)$ respectively. The mediocre prize $w$ is identified with $(-1,0)$. The intersection point of the indifference curves is $x$. The slopes of the lines $\overleftrightarrow{x\bar{w}}$ and $\overleftrightarrow{x\underline{w}}$ are $\tilde{m}_1$ and $\tilde{m}_0$ respectively. The lottery $p_a$ equals $a\bar{w}+(1-a)\underline{w}$, and the slope of its indifference curve equals $a\tilde{m}_1+(1-a)\tilde{m}_0$.\par}
\end{minipage}
\label{RIEU6}
\end{figure}

Suppose that $W=\{w,\underline{w},\bar{w}\}$, where $\bar{w}$ is the best prize, and $\underline{w}$ the worst. As shown in Figure~\ref{RIEU6}, identify $\bar{w}$ and $\underline{w}$  with the points $(0,1)$ and $(0,-1)$ in $\mathbb{R}^2$ respectively, and let the point $(-1,0)$ denote $w$. If a weighted utility preference is monotone with respect to first-order stochastic dominance, the intersection point of the indifference curves must lie in the green or blue area in the figure. If it lies in the green area, the preference increases counterclockwise. If it lies in the blue area, the preference increases clockwise.

We can depict the possible locations of the intersection point in the following way. Draw a line passing $\underline{w}$ with slope $m_0$, and another line passing $\bar{w}$ with slope $m_1$. The intersection of these two lines lies in the green or blue area if and only if $(m_0,m_1)\in[-1,1]^2$. Therefore, a random weighted utility with monotonicity in first-order stochastic dominance can be identified with a random vector $(\tilde{m}_0,\tilde{m}_1)\in\mathbb{R}^2$ whose support is $[-1,1]^2$. 

For any $a\in(0,1)$, let $p_a=a\bar{w}+(1-a)\underline{w}$. The slope of the indifference curve of $p_a$ is equal to $a\tilde{m}_1+(1-a)\tilde{m}_0$ (see Figure~\ref{RIEU6}). The probability that $p_a$ is chosen over another lottery $q$ is equal to the probability that the slope of the indifference curve of $p_a$ is less than the slope of $\overleftrightarrow{p_aq}$ (i.e.\ $q$ lies below the indifference curve of $p_a$). Therefore, we can recover the distribution of 
$a\tilde{m}_1+(1-a)\tilde{m}_0$ for any $a\in(0,1)$ from random choice. This allows us to pin down all the joint moments of $\tilde{m_0}$ and $\tilde{m_1}$.

To find out the moments, fix any natural number $n$. Take any $n+1$ distinct $b_0,\cdots,b_n\in(0,1)$. For each $k\in\{0,\cdots,n\}$, we can compute $E[(\tilde{m}_1+b_k\tilde{m}_0)^n]$. Observe that $$E[(\tilde{m}_1+b_k\tilde{m}_0)^n]=E[\tilde{m}_1^n]+\cdots+{n\choose j}b_k^{j}E[\tilde{m}_1^{n-j}\tilde{m}_0^j]+\cdots+b_k^nE[\tilde{m}_0^n].$$ Hence we have a system of linear equations for the unknown moments $E[\tilde{m}_1^{n-j}\tilde{m}_0^j]$, $0\leq j\leq n$: 
\[
\begin{bmatrix} 
    1 & \cdots & {n\choose j}b_0^{j} & \cdots & b_0^n\\
    \vdots & \ddots & & & \vdots\\
    \vdots & \cdots & {n\choose j}b_k^{j} & \cdots & b_k^n\\
    \vdots &  & & \ddots & \vdots\\
    1 & \cdots & {n\choose j}b_n^{j} & \cdots & b_n^n
\end{bmatrix}\times
\begin{bmatrix}
E[\tilde{m}_1^n] \\ \vdots \\ E[\tilde{m}_1^{n-j}\tilde{m}_0^j] \\ \vdots \\ E[\tilde{m}_0^n]
\end{bmatrix}=
\begin{bmatrix}
E[(\tilde{m}_1+b_0\tilde{m}_0)^n] \\ \vdots \\ E[(\tilde{m}_1+b_k\tilde{m}_0)^n] \\ \vdots \\ E[(\tilde{m}_1+b_n\tilde{m}_0)^n]
\end{bmatrix}.
\]
The $(n+1)$-by-$(n+1)$ matrix on the left-hand side is equivalent to a square Vandermonde matrix. It is invertible as long as all $b_0,\cdots,b_n$ are distinct. Thus this system of linear equations has a unique solution. Consequently, random choice pins down $E[\tilde{m}_1^i\tilde{m}_0^j]$ for any nonnegative integers $i$ and $j$.

Now the identification problem translates into a classic moment problem: if all the joint moments are given, is there a unique joint distribution that generates these moments? The answer is yes in our case because the support of $(\tilde{m}_0,\tilde{m}_1)$ is compact. A distribution with a compact support is uniquely determined by its moments.\footnote{It can be proved by the Stone-Weierstrass theorem, which implies that any continuous function on a compact subset of a multi-dimensional Euclidean space can be uniformly approximated by polynomials.} Note that without first-order stochastic dominance, the support would not be bounded.  So the argument is not valid for general random weighted utility, and as shown in Section 3, uniqueness fails.\footnote{In general, two different distributions can generate identical moments. The moment problem has been studied for decades. See \cite{Kleiber:2013} for a recent study.} Our result here shows that stochastic dominance may help restoring unique identification when risk preferences do not conform to expected utility theory.

\begin{prop}
Suppose that there are three prizes, and the best and worst ones are prespecified. Then a random choice correspondence is rationalized by at most one random weighted utility under which all possible risk preferences are monotonic with respect to first-order stochastic dominance.
\end{prop}

This result is valid for local random choice data. The above identification strategy works if we observe random choice from menus $\{p_a,q\}$  for all $a$ in an arbitrary open subset of $[0,1]$ and all $q$ in a neighborhood of $p_a$. Note that the identification of REU has a similar property. Fixing any $p$, random choice from menu $\{p,q,r\}$ for all $q,r$ in a neighborhood of $p$ is sufficient to pin down REU.  

\subsection{Random joint choice across menus}
The lack of unique identification implies that random choice does not fully capture the implications of a random utility model. Sometimes we can add more assumptions into a model to improve identification. However, a modeler may find those assumptions unappealing, and the actual choice data may not be consistent with a more constrained model. Another way to restore uniqueness, as demonstrated in this section, is to enrich the observable.  
\par
Under the classic notion of random choice, the distribution of choice out of each menu is observable. Although there are infinitely many menus in our setting, random choice is by no means rich choice data in general. To see this, we first review how the identification of REU works. In contrast to Proposition 1, when all possible preferences conform to expected utility theory then random choice identifies a unique distribution of preferences.
\begin{thm}[\citealp{Gul:2006}]
A random choice correspondence is rationalized by at most one random expected utility.
\end{thm}
\par
We suggest the reason for this difference between REU and RIEU. The key  concerns the {\it joint distribution} of choice across two or more menus. Such information is implicitly revealed through random choice under REU, but not under RIEU. A brief proof of Theorem 1 demonstrates this point.
\par
\begin{proof}[Sketch of Proof]
We need some preliminaries. First, define REU formally. Let $\Omega_E$ denote the set of all expected utility preferences over $\Delta$. Let $$N_E(D,A)\coloneqq\{\succsim\in\Omega_E:A=M(D,\succsim)\}.$$ Then let $$\mathcal{C}_E\coloneqq\{N_E(D,A):A,D\in\mathcal{D}\},$$ and let $\mathcal{F}(\mathcal{C}_E)$ denote the smallest field that contains every element of $\mathcal{C}_E$. A random expected utility (REU) is a finitely additive probability measure $\mu_E$ on $(\Omega_E,\mathcal{F}(\mathcal{C}_E))$.
\par
A class $\mathcal{A}$ of subsets of a set $X$ is called a {\it semiring} if (i) $\emptyset\in\mathcal{A}$; (ii) if $A,B\in\mathcal{A}$ then $A\cap B\in\mathcal{A}$; (iii) if $A,B\in\mathcal{A}$ then $A\setminus B=\cup_{k=1}^mC_k$ for some mutually disjoint sets $C_1,\cdots,C_m\in\mathcal{A}$. Say that $\mathcal{A}$ is a {\it semifield} if it is a semiring and $X\in\mathcal{A}$. A set function is a mapping $\nu:\mathcal{A}\rightarrow\mathbb{R}$. The set function $\nu$ is {\it finitely additive} if $\nu(\cup_{i=1}^mA_i)=\sum_{i=1}^m\nu(A_i)$ for any finite collection of mutually disjoint sets $\{A_1,\cdots,A_m\}\subset\mathcal{A}$ such that $\cup_{i=1}^mA_i\in\mathcal{A}$.
\par
The argument for the unique identification of REU is the following. Given RCC $\rho$, define a set function $\mu^*_E:\mathcal{C}_E\rightarrow[0,1]$ such that $\mu^*_E(N_E(D,A))=\rho_D(A)$ for all $N_E(D,A)\in\mathcal{C}_E$. Suppose that $\rho$ is rationalized by REU. Then $\mu^*_E$ is a finitely additive set function. Since $\mathcal{C}_E$ is a semifield, the extension theorem \citep[Theorem 3.5.1]{Rao:1983} applies, and the set function $\mu^*_E$ can be extended to a finitely additive measure $\mu_E$ on $\mathcal{F}(\mathcal{C}_E)$. Moreover, such extension is unique. Hence $\mu_E$ is the unique REU that rationalizes $\rho$.
\end{proof}
The key to the uniqueness of REU is that $\mathcal{C}_E$ is a semifield. This ensures that the extension of $\mu^*_E$ is unique. To identify RIEU, define $\mathcal{C}$ by (\ref{C}) and then define a set function $\mu^*:\mathcal{C}\rightarrow[0,1]$ such that, for all $N(D,A)\in\mathcal{C}$, $\mu^*(N(D,A))=\rho_D(A)$. RIEU $\mu$ rationalizes $\rho$ if and only if it is an extension of $\mu^*$. However, $\mathcal{C}$ is not a semifield. Hence the extension of $\mu^*$ may not be unique. If, on $(\Omega,\mathcal{F}(\mathcal{C}))$, there exist two different probability measures $\mu$ and $\mu'$ that both agree with $\mu^*$ on $\mathcal{C}$, then both of them are RIEUs and rationalize $\rho$.
\par
The class $\mathcal{C}$ fails to be a semifield because it is not closed under finite intersections. That is, even if $N(D,A)$ and $N(D',A')$ are both in $\mathcal{C}$, $N(D,A)\cap N(D',A')$ may be not. This reflects the fact that, {\it under RIEU, random choice does not reveal the joint distribution of choice across two menus}. That is, $\rho$ does not give the joint probability that $A$ is the set of all optimal lotteries in $D$ {\it and} $A'$ is the set of all optimal lotteries in $D'$.
\par
For example, consider RIEUs $\mu$ and $\mu'$ defined in Section 3.1. Let $p=(0,\frac{1}{2})$, $q=(\frac{1}{4},\frac{3}{4})$, $p'=(\frac{1}{2},0)$ and $q'=(\frac{3}{4},\frac{1}{4})$. They are depicted in Figures~\ref{RIEU1} and \ref{RIEU2}. Note that $p\succ_1(\succ'_1)q$ and $p'\prec_1(\succ'_1)q'$, and that $p\prec_2(\prec'_2)q$ and $p'\succ_2(\prec'_2)q'$. Therefore, under $\mu$, the probability that $p$ is chosen over $q$ {\it and} $p'$ is chosen over $q'$ is 0. However, under $\mu'$, this probability equals $\frac{1}{2}$. Thus, $\mu$ and $\mu'$ disagree on the joint distribution of choice across $\{p,q\}$ and $\{p',q'\}$.\footnote{Similarly, in our second example, $v_1$ and $v_2$ could be distinguished if we were given joint distributions of choice across two menus. Consider menus $\{p,q\}$ and $\{p',q'\}$, where $p'=\frac{1}{2}p+\frac{1}{2}r$ and $q'=\frac{1}{2}q+\frac{1}{2}r$ for some $r$. Under $\nu_1$, the probability of $p$ being chosen from $\{p,q\}$ and $q'$ being chosen from $\{p',q'\}$ is positive. Under $\nu_2$, that probability is zero. In other words, if $\nu_1$ is the actual distribution of preferences, then with positive probability, the joint choice will violate the independence axiom .}
\par
On the other hand, the class $\mathcal{C}_E$ is closed under finite intersections because, for any $\lambda\in(0,1)$,\footnote{Define mixtures of menus $D$ and $D'$ by $\lambda D+(1-\lambda)D'\coloneqq\{\lambda p+(1-\lambda)p':p\in D,p'\in D'\}$ for any $\lambda\in(0,1]$.} $$N_E(D,A)\cap N_E(D',A')=N_E(\lambda D+(1-\lambda)D',\lambda A+(1-\lambda)A')\in\mathcal{C}_E.$$ This follows from the independence axiom of expected utility. Specifically, for any expected utility preference $\succsim_E$, $\{p,q\}\subset D$, $\{p',q'\}\subset D'$, and $\lambda\in(0,1)$, \begin{align*}p\succsim_E q&\Leftrightarrow\lambda p+(1-\lambda)p'\succsim_E\lambda q+(1-\lambda)p';\\p'\succsim_E q'&\Leftrightarrow\lambda p+(1-\lambda)p'\succsim_E\lambda p+(1-\lambda)q'.\end{align*} These imply that $\lambda p+(1-\lambda)p'$ is optimal in $\lambda D+(1-\lambda)D'$ if and only if $p$ and $p'$ are optimal in $D$ and $D'$ respectively. In words, under REU, the distribution of choice from the mixture of $D$ and $D'$ reveals the joint distribution of choice across $D$ and $D'$. Once we deviate from REU, we may lose such information from random choice.
\par
For a given random utility model, a natural question is what is the minimal requirement on choice data for unique identification. Or at least, one would like to know how to enrich data to narrow down the set of rationalizing random preferences. The above discussion suggests that collecting data on joint frequencies of choice across different menus may be helpful. 
\par
Regardless of the domain of preferences, if the joint distribution of choice across any number of menus is observable, then the distribution of preferences is uniquely recoverable. It is because the class $$\{\cap_{i=1}^kN(D_i,A_i):A_i,D_i\in\mathcal{D}\ \forall\ 1\leq i\leq k;k\geq1\}$$ is closed under finite intersection. Such class is called a $\pi$-system. If two probability measures agree on a $\pi$-system that generates the field where they are defined, then they are identical.\footnote{This follows from Dynkin's $\pi$-$\lambda$ theorem.}
\par
However, it seems far-fetched to assume that all joint distributions of choice are observable. An advantage of REU is that no joint choice distribution is needed to identify the distribution of preferences. If we deviate from REU but still consider a restrictive domain of preferences, then joint distributions of choice across a limited number of menus might be sufficient. Our last result for random weighted utility demonstrates this point.
\par
What are the minimal observations to completely determine a weighted utility preference on the Marschak-Machina triangle? If we know $p\sim q$, $p'\sim q'$, and $p''\succ q''$, then we can infer the entire indifference map. This is because the lines $\overleftrightarrow{pq}$ and $\overleftrightarrow{p'q'}$ determine the intersection point of all the indifference curves, and then $p''\succ q''$ determines the direction of the increasing preference. A similar result holds for random weighted utility: Joint distributions of choice across three binary menus pin down a unique distribution of weighted utility preferences.
\par
\begin{prop}
Suppose that there are three prizes. Suppose that the joint distribution of choice across any three binary menus is observable. Then such behavior is rationalizable by at most one random weighted utility.
\end{prop}
\begin{proof}
Now $\Omega$ denotes the set of all weighted utility preferences. The class $\mathcal{C}$ is still defined by (\ref{C}), and $\mathcal{F}(\mathcal{C})$ is the smallest field generated by $\mathcal{C}$. 
\par
As shown in Appendix C, if two random weighted utilities $\mu$ and $\mu'$ agree on the class $$\{\cap_{i=1}^3N(\{p_i,q_i\},A_i):A_i\subset\{p_i,q_i\}\subset\Delta\ \forall\ i=1,\cdots,3\},$$ then they also agree on the class $$\mathcal{E}\coloneqq\{\cap_{i=1}^kN(\{p_i,q_i\},A_i):A_i\subset\{p_i,q_i\}\subset\Delta\ \forall\ i=1,\cdots,k; k\geq 1\}.$$ The class $\mathcal{E}$ is a $\pi$-system. Moreover, the smallest field generated by $\mathcal{E}$ is also $\mathcal{F}(\mathcal{C})$. Therefore, $\mu$ and $\mu'$ agree on $\mathcal{F}(\mathcal{C})$.
\end{proof}
\par
The key is that the joint distribution of choice across any four or more binary menus can be deduced from the observable. Figure~\ref{RIEU5} illustrates. There are four pairs of lotteries, $\{(p_i,q_i)\}_{i=1}^4$. If $p_i\succ q_i$, then $\succsim$ rotates at some point (the intersection of all its indifference curves) in one side of the line $\overleftrightarrow{p_iq_i}$ clockwise (the direction of increasing preference), or in another side counterclockwise. Thus, in Figure~\ref{RIEU5}(a), if $p_i\succ q_i$ for all $i$, then $\succsim$ rotates at some point in the blue area clockwise, or in the green area counterclockwise.
\par
In Figure~\ref{RIEU5}(b), we plot three additional lotteries $r$, $s$, and $t$. The lines $\overleftrightarrow{rs}$ and $\overleftrightarrow{rt}$ divide the blue and green areas into five regions. Note that, for instance, if $r\succ s$, $t\succ r$, and $p_4\succ q_4$, then $\succsim$ rotates at some point in the region $R_2$ clockwise, or in $R_5$ counterclockwise. One can verify that \begin{align}p_i\succ q_i\ \forall\ i=1,\cdots,4\Longleftrightarrow&\left(p_1\succ q_1\land p_2\succ q_2\land s\succ r\right)\nonumber\\&\lor\left(r\succsim s\land t\succsim r\land p_4\succ q_4\right)\\&\lor\left(p_1\succ q_1\land r\succ t\land p_3\succ q_3\right).\nonumber\end{align} The right-hand side is a disjunction of three mutually exclusive statements. Let $P_{D_1,D_2,D_3}$ denote the joint distribution of choice across menus $D_1$, $D_2$, and $D_3$. Then (2) implies that \begin{align*}&\mbox{Probability that }p_i\mbox{ is chosen from }\{p_i,q_i\}\mbox{ for all }i=1,\cdots,4\\=&P_{\{p_1,q_1\},\{p_2,q_2\},\{r,s\}}\left(p_1, p_2, s\right)+P_{\{r,s\},\{r,t\},\{p_4,q_4\}}\left(r\lor\{r,s\}, t\lor\{r,t\}, p_4\right)\\&+P_{\{p_1,q_1\},\{r,t\},\{p_3,q_3\}}\left(p_1, r, p_3\right).\end{align*} Therefore, although the joint probability that $p_i$ is chosen over $q_i$ for all $i=1,\cdots,4$ is not observed directly, it is implicitly revealed through the observables. Once we have joint distributions of choice across three binary menus, we automatically get all joint choice probabilities. Hence we can identify a unique distribution of preferences.

\begin{figure}
\centering
\begin{tikzpicture}[line cap=round,line join=round,>=triangle 45,x=3.7cm,y=3.7cm]
\clip(-0.5,-0.3) rectangle (3.5,1.3);
\fill[line width=0pt,color=blue,fill=blue,fill opacity=0.1] (0.6,0.6) -- (1.2,0.6) -- (1.2,0) -- (1,0) -- (0.6,0.4) -- cycle;
\fill[line width=0pt,color=blue,fill=blue,fill opacity=0.1] (0.6,-0.25) -- (1.2,-0.25) -- (1.2,0) -- (0.6,0) -- cycle;
\fill[line width=0pt,color=ForestGreen,fill=ForestGreen,fill opacity=0.1] (0.4,0.8) -- (0.6,1) -- (0.6,1.25) -- (0,1.25) -- (0,1) -- (0.2,0.8) -- cycle;
\fill[line width=0pt,color=ForestGreen,fill=ForestGreen,fill opacity=0.1] (0,0.8) -- (-0.2,0.8) -- (-0.2,1.25) -- (0,1.25) -- cycle;
\fill[line width=0pt,color=blue,fill=blue,fill opacity=0.1] (2.6,0.6) -- (3.2,0.6) -- (3.2,0) -- (3,0) -- (2.6,0.4) -- cycle;
\fill[line width=0pt,color=blue,fill=blue,fill opacity=0.1] (2.6,0) -- (3.2,0) -- (3.2,-0.25) -- (2.6,-0.25) -- cycle;
\fill[line width=0pt,color=ForestGreen,fill=ForestGreen,fill opacity=0.1] (2.4,0.8) -- (2.6,1) -- (2.6,1.25) -- (2,1.25) -- (2,1) -- (2.2,0.8) -- cycle;
\fill[line width=0pt,color=ForestGreen,fill=ForestGreen,fill opacity=0.1] (2,0.8) -- (2,1.25) -- (1.8,1.25) -- (1.8,0.8) -- cycle;
\draw [line width=1pt,color=darkgray] (0,0)-- (1,0);
\draw [line width=1pt,color=darkgray] (1,0)-- (0,1);
\draw [line width=1pt,color=darkgray] (0,1)-- (0,0);
\draw [line width=1pt,color=darkgray] (2,0)-- (3,0);
\draw [line width=1pt,color=darkgray] (3,0)-- (2,1);
\draw [line width=1pt,color=darkgray] (2,1)-- (2,0);
\draw [line width=1pt,dash pattern=on 1pt off 1pt,color=darkgray] (-0.2,0.6)-- (1.2,0.6);
\draw [line width=1pt,dash pattern=on 1pt off 1pt,color=darkgray] (0.6,-0.25)-- (0.6,1.25);
\draw [line width=1pt,dash pattern=on 1pt off 1pt,color=darkgray] (0,0.4)-- (0.6,1);
\draw [line width=1pt,dash pattern=on 1pt off 1pt,color=darkgray] (-0.2,0.8)-- (1,0.8);
\draw [line width=1pt,dash pattern=on 1pt off 1pt,color=darkgray] (1.8,0.6)-- (3.2,0.6);
\draw [line width=1pt,dash pattern=on 1pt off 1pt,color=darkgray] (2.6,-0.25)-- (2.6,1.25);
\draw [line width=1pt,dash pattern=on 1pt off 1pt,color=darkgray] (2,0.4)-- (2.6,1);
\draw [line width=1pt,dash pattern=on 1pt off 1pt,color=darkgray] (1.8,0.8)-- (3,0.8);
\draw [line width=1pt,dash pattern=on 1pt off 1pt,color=darkgray] (2.75,-0.25)-- (2.25,1.25);
\draw [line width=1pt,dash pattern=on 1pt off 1pt,color=darkgray] (1.8,0.2)-- (3.2,0.2);
\draw (-0.5,1.2) node[anchor=north west] {(a)};
\draw (1.5,1.2) node[anchor=north west] {(b)};
\begin{scriptsize}
\draw [color=blue] (0.6,0) circle (2pt);
\draw[color=blue] (0.56,-0.05) node {$q_1$};
\draw [fill=blue] (0.6,0.4) circle (2pt);
\draw[color=blue] (0.67,0.43) node {$p_1$};
\draw [color=blue] (0,0.6) circle (2pt);
\draw[color=blue] (-0.05,0.54) node {$q_2$};
\draw [fill=blue] (0.4,0.6) circle (2pt);
\draw[color=blue] (0.4,0.54) node {$p_2$};
\draw [color=blue] (0,0.8) circle (2pt);
\draw[color=blue] (-0.05,0.74) node {$q_4$};
\draw [fill=blue] (0.2,0.8) circle (2pt);
\draw[color=blue] (0.26,0.84) node {$p_4$};
\draw [color=blue] (0,0.4) circle (2pt);
\draw[color=blue] (-0.05,0.34) node {$q_3$};
\draw [fill=blue] (0.3,0.7) circle (2pt);
\draw[color=blue] (0.24,0.7) node {$p_3$};
\draw [color=blue] (2.6,0) circle (2pt);
\draw[color=blue] (2.56,-0.05) node {$q_1$};
\draw [color=blue] (2,0.4) circle (2pt);
\draw[color=blue] (1.95,0.34) node {$q_3$};
\draw [fill=blue] (2.6,0.4) circle (2pt);
\draw[color=blue] (2.67,0.43) node {$p_1$};
\draw [color=blue] (2,0.6) circle (2pt);
\draw[color=blue] (1.95,0.54) node {$q_2$};
\draw [fill=blue] (2.4,0.6) circle (2pt);
\draw[color=blue] (2.4,0.54) node {$p_2$};
\draw [fill=blue] (2.3,0.7) circle (2pt);
\draw[color=blue] (2.24,0.7) node {$p_3$};
\draw [color=blue] (2,0.8) circle (2pt);
\draw[color=blue] (1.95,0.74) node {$q_4$};
\draw [fill=blue] (2.2,0.8) circle (2pt);
\draw[color=blue] (2.26,0.84) node {$p_4$};
\draw [fill=blue] (2.6,0.2) ++(-2pt,0 pt) -- ++(2pt,2.5pt)--++(2pt,-2pt)--++(-2pt,-2pt)--++(-2pt,2pt);
\draw[color=blue] (2.57,0.16) node {$r$};
\draw [fill=blue] (2,0.2) ++(-2pt,0 pt) -- ++(2pt,2pt)--++(2pt,-2pt)--++(-2pt,-2pt)--++(-2pt,2pt);
\draw[color=blue] (1.97,0.16) node {$s$};
\draw [fill=blue] (2.5,0.5) ++(-2pt,0 pt) -- ++(2pt,2pt)--++(2pt,-2pt)--++(-2pt,-2pt)--++(-2pt,2pt);
\draw[color=blue] (2.47,0.46) node {$t$};
\draw (2.9,0.43) node[anchor=north west] {$R_1$};
\draw (2.9,-0.08) node[anchor=north west] {$R_2$};
\draw (2.58,-0.1) node[anchor=north west] {$R_3$};
\draw (2.4,1.1) node[anchor=north west] {$R_4$};
\draw (2,1.15) node[anchor=north west] {$R_5$};
\end{scriptsize}
\end{tikzpicture}
\caption{}
\medskip
\begin{minipage}{0.75\textwidth}
{\scriptsize (a): Suppose that $\succsim$ is a weighted utility preference. Then $p_i\succ q_i$ for all $i=1,\cdots,4$ if and only if $\succsim$ rotates at some point in the blue area clockwise or rotates at some point in the green area counterclockwise. \\(b): Suppose that $\succsim$ is a weighted utility preference. Then $p_1\succ q_1$, $p_2\succ q_2$, and $s\succ r$ if and only if $\succsim$ rotates at some point in the area $R_1$; $r\succ s$, $t\succ r$, and $p_4\succ q_4$ if and only if $\succsim$ rotates at some point in the area $R_2$ or $R_5$; $p_1\succ q_1$, $r\succ t$, and $p_3\succ q_3$ if and only if $\succsim$ rotates at some point in the area $R_3$ or $R_4$.\par}
\end{minipage}
\label{RIEU5}
\end{figure}

\section{Conclusions}
This paper addresses the identification of random utility in a three-prize setting of choice under risk. When risk preferences conform to the expected utility theory, the distribution of preferences is uniquely recoverable from random choice of lotteries. But such uniqueness fails in general if risk preferences deviate from expected utility. We show by example that even if preferences are confined to a specific class of non-expected utility, such as the betweenness class or a subclass, non-uniqueness can obtain. On the other hand, unique identification is not hopeless at all. We find that if risk preferences conform to the weighted utility theory and are monotone in first-order stochastic dominance, then the distribution of preferences is again uniquely recoverable from random choice. In general, collecting data on joint choice distributions across different menus can improve identification. If preferences are suitably restricted, random joint choice across a small number of menus will suffice for unique identification. 
\par
The paper is illustrative but does not dive into any particular model of random non-expected utility. Many classes of non-expected utility are not considered in this paper, such as rank-dependent expected utility \citep{Quiggin:1982}, disappointment-averse utility \citep{Gul:1991}, or cautious expected utility \citep{Cerreia-Vioglio:2015}. A random utility modeler may choose a domain of preferences based on what describes each individual's behavior better in her belief. But she may also rather select a model with a unique identification. This paper points out a generic difficulty a modeler may face and suggests some possible solutions. Nonetheless, a comprehensive study of random non-expected utility is beyond the scope of this paper and is left for future research.

\appendix
\appendixpage
\section{Representation of Betweenness Preference}
A betweenness preference has the following representation.
\begin{prop} [Implicit Expected Utility Representation]
A preference over $\Delta$ is a betweenness preference if and only if there exists $u(\cdot,\cdot):W\times[0,1]\rightarrow\mathbb{R}$, continuous in the second argument, such that $p\succsim q\Leftrightarrow V(p)\geq V(q)$, where $V(p)$ is defined implicitly as the unique $v\in[0,1]$ that solves \begin{equation}\sum_{i=1}^{n+1}u(w_i,v)p^i=vu(\bar{w},v)+(1-v)u(\underline{w},v).\label{R}\end{equation} Furthermore, $u(w,v)$ is unique up to positive affine transformations which are continuous in $v$. A particular transformation exists setting $u(\underline{w},v)=0$ and $u(\bar{w},v)=1$ for all $v\in[0,1]$.
\end{prop}
\begin{proof}
See \citet[Proposition A.1]{Dekel:1986}.
\end{proof}

\section{Behavioral Properties of RIEU}
For any $D,D'\in\mathcal{D}$ and $\lambda,\lambda'\geq0$, let $\lambda D+\lambda'D'\coloneqq\{\lambda x+\lambda'y:x\in D,y\in D'\}$. Note that $\lambda D+\lambda'D'$ is also a menu.
\par
For any convex set $C$, a convex set $F\subset C$ is called a {\it face} of $C$ if for all $x,y\in C$ and $\lambda\in(0,1)$, $$\lambda x+(1-\lambda)y\in F\Rightarrow\{x,y\}\subset F.$$
For any set $A$, let $chA$ denote the convex hull of $A$. For any $D,A\in\mathcal{D}$, we say that $A$ is a face of $D$ if $chA$ is a face of $chD$ and $chA\cap D=A$.
\par
\begin{axm}
Monotonicity: $\rho_D(A)\leq\rho_{D\setminus B}(A\setminus B)$ for all $D,A,B\in\mathcal{D}$ with $A\setminus B\neq\emptyset$.
\end{axm}
Monotonicity captures the intuition that the probability of lottery $p$ being optimal does not decrease as some lotteries are removed from the menu. This property holds under any random utility model.
\begin{axm}
Extremeness: $\rho_D(A)>0$ implies that $A$ is a face of $D$, for all $A,D\in\mathcal{D}$.
\end{axm}
Extremeness states that lotteries $p$ and $q$ are both optimal whenever a mixture of them is optimal. This property reflects the fact that each indifference set of a betweenness preference is linear.
\begin{axm}
Stochastic Betweenness: $\rho_{\lambda D+(1-\lambda)p}(\lambda A+(1-\lambda)p)=\rho_D(A)$ for all $p\in A\subset D\in\mathcal{D}$, $\lambda\in(0,1)$.
\end{axm}
Stochastic Betweenness states that the probability of lottery $p$ being optimal remains unchanged when each other lottery $q$ in the menu is replaced by a mixture of $p$ and $q$. This property captures the betweenness axiom, which requires that a mixture of two lotteries should lie in between them in preference.

\begin{prop}
If $\rho$ is rationalized by RIEU, then $\rho$ satisfies Monotonicity, Extremeness, and Stochastic Betweenness.
\end{prop}
\begin{proof}
Suppose that $\rho$ is rationalized by RIEU $\mu$.
\par
Monotonicity: If $A=M(D,\succsim)$, then $A\setminus B=M(D\setminus B,\succsim)$. Thus $N(D,A)\subset N(D\setminus B,A\setminus B)$. It follows that $\rho_D(A)=\mu(N(D,A))\leq\mu(N(D\setminus B,A\setminus B))=\rho_{D\setminus B}(A\setminus B)$. Hence $\rho$ satisfies Monotonicity.
\par
Extremeness: Suppose that $A=M(D,\succsim)$, where $\succsim$ is a betweenness preference. Betweenness property implies that $p\sim p'$ for all $p,p'\in chA$, and $p\succ p''$ for all $p\in chA$ and $p''\in chD\setminus ch A$. Thus $chA=M(chD,\succsim)$. Suppose that $p\in chA$ and $p=\lambda q+(1-\lambda)r$ where $\lambda\in(0,1)$ and $\{q,r\}\subset chD$. Since $p$ is optimal, $p\succsim q$ and $p\succsim r$. If $p\succ q$ or $p\succ r$, then by betweenness property, $p\succ\lambda q+(1-\lambda)r=p$, a contradiction. Hence $p\sim q\sim r$, implying that $q$ and $r$ are both in $ch A$. Thus $chA$ is a face of $chD$. Since $A=M(D,\succsim)$, $chA\cap D=A$. Thus $A$ is a face of $D$ and so $\rho$ satisfies Extremeness.
\par
Stochastic Betweenness: Suppose that $p\in A=M(D,\succsim)$, where $\succsim$ is a betweenness preference.  Then for all $q\in A$ and $r\in D\setminus A$, $p\sim q$ and $p\succ r$. By betweenness property, $p\sim\lambda q+(1-\lambda)p$ and $p\succ\lambda r+(1-\lambda)p$ for all $\lambda\in(0,1)$. Hence $N(D,A)\subset N(\lambda D+(1-\lambda)p,\lambda A+(1-\lambda)p)$. The argument can be reversed to conclude that $N(D,A)\supset N(\lambda D+(1-\lambda)p,\lambda A+(1-\lambda)p)$. Thus $\rho_{\lambda D+(1-\lambda)p}(\lambda A+(1-\lambda)p)=\rho_D(A)$. So $\rho$ satisfies Stochastic Betweenness.
\end{proof}

\section{Proof of Proposition 4}
We want to show that if two random weighted utilities $\mu$ and $\mu'$ agree on the class $$\mathcal{E}_0\coloneqq\{\cap_{i=1}^3N(\{p_i,q_i\},A_i):A_i\subset\{p_i,q_i\}\subset\Delta\ \forall\ i=1,\cdots,3\},$$ then they also agree on the class $$\mathcal{E}\coloneqq\{\cap_{i=1}^kN(\{p_i,q_i\},A_i):A_i\subset\{p_i,q_i\}\subset\Delta\ \forall\ i=1,\cdots,k; k\geq 1\}.$$ This would be true if any $E\in\mathcal{E}$ can be expressed as $E=\cup_{m=1}^ME_m$, where $E_m\in\mathcal{E}_0$ for all $m$, and $E_m\cap E_{m'}=\emptyset$ for all $m\neq m'$.
\par
Consider four pairs of lotteries, $(p_i,q_i)$ for $i=1,\cdots,4$. Let $L_i$ denote the line $\overleftrightarrow{p_iq_i}$. Each $L_i$ divides $\mathbb{R}^2$ into two half spaces, $H_i^+$ and $H_i^-$. Without loss of generality, assume that $p_i\succsim q_i$ if $\succsim$ rotates at some point in $H_i^+$ clockwise or rotates at some point in $H_i^-$ counterclockwise.
\par
We want to show that the set $\{\succsim:p_i\succsim q_i\ \forall\ i=1\cdots,4\}$ can be decomposed into finitely many mutually disjoint subsets. Moreover, each subset takes the form $\{\succsim:r_j\succ s_j\lor r_j\sim s_j\ \forall\ j=1\cdots,3\}$. By induction, we can extend the claim to the case of more than four pairs of lotteries.
\par
It suffices to show the following: There exist $r_{jk},s_{jk}$ for $j\in\{1,2,3\}$ and $k\in\{1,\cdots,K\}$ such that $$\cap_{i=1}^4\{\succsim:p_i\succsim q_i\}=\cup_{k=1}^K\left(\cap_{j=1}^3\{\succsim:r_{jk}\succsim s_{jk}\}\right),$$ and, for all $k\neq k'$, $$\left(\cap_{j=1}^3\{\succsim:r_{jk}\succsim s_{jk}\}\right)\cap\left(\cap_{j=1}^3\{\succsim:r_{jk'}\succsim s_{jk'}\}\right)\subset\{\succsim:r\sim s\}$$ for some $r,s\in\Delta$.
\par
If $L_i=L_j$ for some $i\neq j$, then the claim is trivial. Now assume that $L_i\neq L_j$ for all $i\neq j$. Without loss of generality, assume that $\cap_{i=1}^4H_i^+\neq\emptyset$.
\par
\noindent {\bf Case 1:} $\cap_{i=1}^4H_i^+$ has only 1 one-dimensional face. That is, it is a half space.
\par
Without loss of generality, assume that $H_1^+\subset\cdots\subset H_4^+$. Then $H_1^-\supset\cdots\supset H_4^-$. Suppose that $p_1\succsim q_1$ and $p_4\succsim q_4$. Then $\succsim$ either rotates at some point in $H_1^+$ clockwise or rotates at some point in $H_4^-$ counterclockwise. Thus $p_2\succsim q_2$ and $p_3\succsim q_3$. We have $$\cap_{i=1}^4\{\succsim:p_i\succsim q_i\}=\{\succsim:p_1\succsim q_1\}\cap\{\succsim:p_4\succsim q_4\}.$$
\par
\noindent {\bf Case 2:} $\cap_{i=1}^4H_i^+$ has two one-dimensional faces. Without loss of generality, assume that $L_1\nparallel L_2$ and $\cap_{i=1}^4H_i^+=H_1^+\cap H_2^+$. 
\par
Case 2-1: $L_1\parallel L_2$. Then $H_1^-\cap H_2^-=\emptyset$. If $p_1\succsim q_1$ and $p_2\succsim q_2$, then $\succsim$ rotates at some point in $H_1^+\cap H_2^+$ clockwise. Thus $p_3\succsim q_3$ and $p_4\succsim q_4$. We have $$\cap_{i=1}^4\{\succsim:p_i\succsim q_i\}=\{\succsim:p_1\succsim q_1\}\cap\{\succsim:p_2\succsim q_2\}.$$
\par
Case 2-2: $L_1\nparallel L_2$ and $L_3\parallel L_4$. Without loss of generality, assume $H_3^+\subset H_4^+$. Then $\cap_{i=1}^4H_i^+=H_1^+\cap H_2^+\cap H_4^+$ and $\cap_{i=1}^4H_i^-=H_1^-\cap H_2^-\cap H_4^-$. Thus $$\cap_{i=1}^4\{\succsim:p_i\succsim q_i\}=\{\succsim:p_1\succsim q_1\}\cap\{\succsim:p_2\succsim q_2\}\cap\{\succsim:p_4\succsim q_4\}.$$
\par
Case 2-3: $L_1\nparallel L_2$, $L_3\nparallel L_4$, and $L_1\cap L_2\subset\Delta$. 
If $L_3\cap L_4\not\subset int(H_1^-\cap H_2^-)$, then either $L_3$ or $L_4$ is redundant. That is, without loss of generality, $\cap_{i=1}^4H_i^+=H_1^+\cap H_2^+\cap H_3^+$ and $\cap_{i=1}^4H_i^-=H_1^-\cap H_2^-\cap H_3^-$. Thus $$\cap_{i=1}^4\{\succsim:p_i\succsim q_i\}=\{\succsim:p_1\succsim q_1\}\cap\{\succsim:p_2\succsim q_2\}\cap\{\succsim:p_3\succsim q_3\}.$$
\par
Now, consider $L_3\cap L_4\subset int(H_1^-\cap H_2^-)$. Let $r\in\Delta$ denote the intersection of $L_1$ and $L_2$. Pick $s\in\Delta$ such that $\overleftrightarrow{rs}$ passes $L_3\cap L_4$. Note that $\overleftrightarrow{rs}$ divides $H_1^+\cap H_2^+$ into two fans. Let $H_5^+$ and $H_5^-$ denote the half spaces generated by $\overleftrightarrow{rs}$ such that $H_1^+\cap H_2^+=(H_1^+\cap H_5^+)\cup(H_5^-\cap H_2^+)$. Without loss of generality, assume that $r\succsim s$ if $\succsim$ rotates at some point in $H_5^+$ clockwise and that $H_5^-\cap H_3^-\subset H_5^-\cap H_4^-$. Then, \begin{align*}\left(\cap_{i=1}^4H_i^+\right)\cup\left(\cap_{i=1}^4H_i^-\right)=&\left[(H_1^+\cap H_3^+\cap H_5^+)\cup(H_1^-\cap H_3^-\cap H_5^-)\right]\\&\cup\left[(H_2^+\cap H_4^+\cap H_5^-)\cup(H_2^-\cap H_4^-\cap H_5^+)\right].\end{align*} This implies that $$\cap_{i=1}^4\{p_i\succsim q_i\}=\left(\{p_1\succsim q_1\}\cap\{p_3\succsim q_3\}\cap\{r\succsim s\}\right)\cup\left(\{p_2\succsim q_2\}\cap\{p_4\succsim q_4\}\cap\{r\precsim s\}\right)$$
\par
Case 2-4: $L_1\nparallel L_2$, $L_3\nparallel L_4$, and $L_1\cap L_2\not\subset int\Delta$. Without loss of generality, assume that $\{p_2,q_2\}\subset H_1^+$. As in the previous case, we shall consider $L_3\cap L_4\subset int(H_1^-\cap H_2^-)$.
\par
Pick a lottery $r$ that is a strict mixture of $p_2$ and $q_2$. Pick $s_1\in\Delta$ such that $\overleftrightarrow{rs_1}\parallel L_1$. Let $H_5^+$ and $H_5^-$ denote the half spaces generated by $\overleftrightarrow{rs_1}$ such that $H_5^+\subset H_1^+$. Without loss of generality, assume that $r\succsim s_1$ if $\succsim$ rotates at some point in $H_5^+$ clockwise. 
\par
Pick $s_2\in\Delta$ such that $\overleftrightarrow{rs_2}$ passes $L_3\cap L_4$. Let $H_6^+$ and $H_6^-$ denote the half spaces generated by $\overleftrightarrow{rs_2}$ such that $H_5^+\cap H_2^+=(H_5^+\cap H_6^+)\cup(H_6^-\cap H_2^+)$. Without loss of generality, assume that $r\succsim s_2$ if $\succsim$ rotates at some point in $H_6^-$ clockwise and that $H_6^+\cap H_4^-\subset H_6^+\cap H_3^-$.
\par
Then \begin{align*}\left(\cap_{i=1}^4H_i^+\right)\cup\left(\cap_{i=1}^4H_i^-\right)=&\left[(H_1^+\cap H_2^+\cap H_5^-)\cup(H_1^-\cap H_2^-\cap H_5^+)\right]\\&\cup\left[(H_2^+\cap H_4^+\cap H_6^-)\cup(H_2^-\cap H_4^-\cap H_6^+)\right]\\&\cup\left[(H_1^+\cap H_3^+\cap H_5^+\cap H_6^+)\cup(H_1^-\cap H_3^-\cap H_5^-\cap H_6^-)\right].\end{align*} This implies that \begin{align*}\cap_{i=1}^4\{p_i\succsim q_i\}=&\left(\{p_1\succsim q_1\}\cap\{p_2\succsim q_2\}\cap\{r\precsim s_1\}\right)\\&\cup\left(\{p_2\succsim q_2\}\cap\{p_4\succsim q_4\}\cap\{r\succsim s_2\}\right)\\&\cup\left(\{p_1\succ q_1\}\cap\{p_3\cap q_3\}\cap\{r\succsim s_1\}\cap\{r\precsim s_2\}\right)\end{align*} Note that $\{p_1\succ q_1\}\cap\{p_3\cap q_3\}\cap\{r\succsim s_1\}\cap\{r\precsim s_2\}$ can be decomposed as in Case 2-3.
\par
\noindent{\bf Case 3:} $\cap_{i=1}^4H_i^+$ has three one-dimensional faces. Without loss of generality, assume that $\cap_{i=1}^4H_i^+=H_1^+\cap H_2^+\cap H_3^+$, $L_1\nparallel L_2$, and $L_2\nparallel L_3$.
\par
Case 3-1: $L_1\parallel L_3$ or $\cap_{i=1}^3H_i^+$ is bounded. Note that $\cap_{i=1}^3H_i^-=\emptyset$. Thus $$\cap_{i=1}^4\{p_i\succsim q_i\}=\{p_1\succsim q_1\}\cap\{p_2\succsim q_2\}\cap\{p_3\succsim q_3\}.$$
\par
Case 3-2: $L_1\nparallel L_3$, $L_1\cap L_3\subset H_2^-\cap H_4^-$. Note that $H_1^-\cap H_3^-\subset H_4^-$. Thus, $$\left(\cap_{i=1}^4H_i^+\right)\cup\left(\cap_{i=1}^4H_i^-\right)=\left(\cap_{i=1}^3H_i^+\right)\cup\left(\cap_{i=1}^3H_i^-\right).$$ Therefore, $$\cap_{i=1}^4\{p_i\succsim q_i\}=\{p_1\succsim q_1\}\cap\{p_2\succsim q_2\}\cap\{p_3\succsim q_3\}.$$
\par
Case 3-3: $L_1\nparallel L_3$, $L_1\cap L_3\subset H_2^-\cap intH_4^+$, and $[(L_1\cap L_2)\cup(L_2\cap L_3)]\cap \Delta\neq\emptyset$.
\par
Without loss of generality, assume that $L_2\cap L_3\equiv\{r\}\subset int\Delta$. Pick $s\in\Delta$ such that $\overleftrightarrow{rs}$ passes the intersection of $L_1$ and $L_4$. Let $H_5^+$ and $H_5^-$ denote the half spaces generated by $\overleftrightarrow{rs}$ such that $\cap_{i=1}^3H_i^+=(H_1^+\cap H_2^+\cap H_5^+)\cup(H_5^-\cap H_3^+)$. Without loss of generality, assume that $r\succsim s$ if $\succsim$ rotates at some point in $H_5^-$ clockwise. Then \begin{align*}\left(\cap_{i=1}^4H_i^+\right)\cup\left(\cap_{i=1}^4H_i^-\right)=&\left[\left(H_1^+\cap H_2^+\cap H_5^+\right)\cup\left(H_1^-\cap H_2^-\cap H_5^-\right)\right]\\&\cup\left[\left(H_3^+\cap H_4^+\cap H_5^-\right)\cup\left(H_3^-\cap H_4^-\cap H_5^+\right)\right].\end{align*} This implies that \begin{align*}\cap_{i=1}^4\{p_i\succsim q_i\}=&\left(\{p_1\succsim q_1\}\cap\{p_2\succsim q_2\}\cap\{s\succsim r\}\right)\\&\cup\left(\{p_3\succsim q_3\}\cap\{p_4\succsim q_4\}\cap\{r\succsim s\}\right).\end{align*}
\par
Case 3-4: $L_1\nparallel L_3$, $L_1\cap L_3\subset H_2^-\cap intH_4^+$, and $[(L_1\cap L_2)\cup(L_2\cap L_3)]\cap \Delta=\emptyset$.
\par
Without loss of generality, assume that $\{p_2,q_2\}\subset H_3^-$. Let $L_5$ denote the line passing $L_2\cap L_3$ and parallel to $L_1$. Note that $L_5$ must intersect $\Delta$ at two points, say $r$ and $s$. One may pick $r$ as a mixture of $p_2$ and $p_1$, $s$ as a mixture of $p_2$ and $q_1$.  Let $H_5^+$ and $H_5^-$ denote the half spaces generated by $L_5$ such that $\cap_{i=1}^3H_i^+=(H_1^+\cap H_2^+\cap H_5^+)\cup(H_5^-\cap H_3^+)$. Without loss of generality, assume that $r\succsim s$ if $\succsim$ rotates at some point in $H_5^-$ clockwise. Then \begin{align*}\left(\cap_{i=1}^4H_i^+\right)\cup\left(\cap_{i=1}^4H_i^-\right)=&\left[\left(H_1^+\cap H_2^+\cap H_5^+\right)\cup\left(H_1^-\cap H_2^-\cap H_5^-\right)\right]\\&\cup\left[\left(H_1^+\cap H_3^+\cap H_4^+\cap H_5^-\right)\cup\left(H_1^-\cap H_3^-\cap H_4^-\cap H_5^+\right)\right].\end{align*} Therefore, \begin{align*}\cap_{i=1}^4\{p_i\succsim q_i\}=&\left(\{p_1\succsim q_1\}\cap\{p_2\succsim q_2\}\cap\{s\succsim r\}\right)\\&\cup\left(\{p_1\succsim q_1\}\cap\{p_3\succsim q_3\}\cap\{p_4\succsim q_4\}\cap\{r\succsim s\}\right).\end{align*} Note that $\{p_1\succsim q_1\}\cap\{p_3\succsim q_3\}\cap\{p_4\succsim q_4\}\cap\{r\succsim s\}$ can be decomposed as in Case 2-4.
\par
\noindent{\bf Case 4:} $\cap_{i=1}^4H_i^+$ has four one-dimensional faces. Without loss of generality, assume that  $L_1\nparallel L_2$; $L_2\nparallel L_3$ and $L_2\cap L_3\subset H_1^+$; $L_3\nparallel L_4$ and $L_3\cap L_4\subset H_2^+$.
\par
Case 4-1: $\cap_{i=1}^4H_i^+$ is bounded; that is, it is a convex polygon with four edges.  Note that, for at least one diagonal, the line containing it intersects $\Delta$ at more than one point. Without loss of generality, pick $r,s\in\Delta$ such that $\overleftrightarrow{rs}$ passes $L_2\cap L_3$ and $L_1\cap L_4$. Let $H_5^+$ and $H_5^-$ denote the half spaces generated by $\overleftrightarrow{rs}$ such that $L_3\cap L_4\subset H_5^+$. Without loss of generality, assume that $r\succsim s$ if $\succsim$ rotates at some point in $H_5^+$ clockwise. Then \begin{align*}\left(\cap_{i=1}^4H_i^+\right)\cup\left(\cap_{i=1}^4H_i^-\right)=&\left[\left(H_1^+\cap H_2^+\cap H_5^-\right)\cup\left(H_1^-\cap H_2^-\cap H_5^+\right)\right]\\&\cup\left[\left(H_3^+\cap H_4^+\cap H_5^+\right)\cup\left(H_3^-\cap H_4^-\cap H_5^-\right)\right].\end{align*} This implies that \begin{align*}\cap_{i=1}^4\{p_i\succsim q_i\}=&\left(\{p_1\succsim q_1\}\cap\{p_2\succsim q_2\}\cap\{s\succsim r\}\right)\\&\cup\left(\{p_3\succsim q_3\}\cap\{p_4\succsim q_4\}\cap\{r\succsim s\}\right).\end{align*}
\par
Case 4-2: $L_1\parallel L_4$. Let $L_5$ be the line that passes $L_2\cap L_3$ and is parallel to $L_1$ and $L_4$. Note that $L_5$ intersects $\Delta$ at more than one point. One may pick $r\in L_5$ as a mixture of $p_1$ and $p_4$, $s\in L_5$ as a mixture of $p_1$ and $q_4$. Then follow the same argument as in Case 4-1.
\par
Case 4-3: $L_1\cap L_4\subset H_2^-\cap H_3^-$. Note that $\cap_{i=1}^4H_i^-=H_1^-\cap H_4^-$, i.e., $\cap_{i=1}^4H_i^-$ has only two one-dimensional faces. We go back to Case 2.

\linespread{1}\selectfont
\bibliographystyle{apa}
\bibliography{RNEU}
\end{NoHyper}\end{document}